\documentclass[11pt]{article} % For LaTeX2e
\usepackage[colorlinks, citecolor={blue}]{hyperref}
\usepackage{amsmath,latexsym,amsbsy,amssymb,amsthm}
\usepackage{url}
\usepackage{color}
\usepackage{graphicx}

\usepackage[margin=1in]{geometry} % RECOMB margins
\usepackage{setspace}
\usepackage{natbib}
\usepackage{caption}

\newtheorem{Theorem}{Theorem}[section]

\newtheorem{Assumption}[Theorem]{Assumption}
\newtheorem{Criterion}[Theorem]{Criterion}

\newtheorem{Lemma}[Theorem]{Lemma}

\numberwithin{equation}{section}

\newcommand{\h}{\hspace*{.24in}}

\title{Online Bayesian phylogenetic inference: \\ theoretical foundations via Sequential Monte Carlo}

\author{Vu Dinh$^1$ \quad
Aaron E. Darling$^2$ \quad
Frederick A. Matsen IV$^1$ \\\\
$^1$Program in Computational Biology\\ Fred Hutchinson Cancer Research Center\\ Seattle, WA, USA\\\\
$^2$The ithree institute\\ University of Technology Sydney\\ Ultimo NSW, Australia \\
}

\date{}

\begin{document}
\maketitle

% RECOMB instructions
% The cover page should contain the title, author(s) names and affiliations, an abstract, the keywords, and the contact author e-mail.
% A manuscript should start with a succinct statement of the problem, the results achieved, their significance, and a comparison with previous work. This material should be understandable to non-specialists. A technical exposition directed to the specialist should follow. The length, excluding cover page and bibliography, should not exceed 10 pages. The manuscript should be easy to read, using at least 11 point font size on U.S. standard 8 1/2 by 11 inch paper with no less than one inch margin all around. If the authors believe that more details are absolutely necessary to substantiate the claims of the paper, they may include a clearly marked appendix, which might be read at the discretion of the reviewers. An email address for the contact author should be included. Manuscripts that deviate significantly from these guidelines risk rejection without consideration of their merits.

\begin{abstract}
Phylogenetics, the inference of evolutionary trees from molecular sequence data such as DNA, is an enterprise that yields valuable evolutionary understanding of many biological systems.
Bayesian phylogenetic algorithms, which approximate a posterior distribution on trees, have become a popular if computationally expensive means of doing phylogenetics.
Modern data collection technologies are quickly adding new sequences to already substantial databases.
With all current techniques for Bayesian phylogenetics, computation must start anew each time a sequence becomes available, making it costly to maintain an up-to-date estimate of a phylogenetic posterior.
These considerations highlight the need for an \emph{online} Bayesian phylogenetic method which can update an existing posterior with new sequences.

Here we provide theoretical results on the consistency and stability of methods for online Bayesian phylogenetic inference based on Sequential Monte Carlo (SMC) and Markov chain Monte Carlo (MCMC).
We first show a consistency result, demonstrating that the method samples from the correct distribution in the limit of a large number of particles.
Next we derive the first reported set of bounds on how phylogenetic likelihood surfaces change when new sequences are added.
These bounds enable us to characterize the theoretical performance of sampling algorithms by bounding the effective sample size (ESS) with a given number of particles from below.
We show that the ESS is guaranteed to grow linearly as the number of particles in an SMC sampler grows.
Surprisingly, this result holds even though the dimensions of the phylogenetic model grow with each new added sequence.

\bigskip



{\bf Funding}: VD and FAM funded by National Science Foundation grants DMS-1223057 and CISE-1564137.
FAM supported by a Faculty Scholar grant from the Howard Hughes Medical Institute and the Simons Foundation.

\end{abstract}

\thispagestyle{empty}
\clearpage
\setcounter{page}{1}

\section{Background and main results}

Phylogenetics is the theory and practice of reconstructing evolutionary trees.
Evolutionary trees have found wide application in biology and medicine, including use in epidemiology, conservation planning, and cancer genomics.
Maximum likelihood and Bayesian methods are generally considered to be the most powerful and accurate approaches for phylogenetic inference.
The Bayesian methods in particular enjoy the flexibility to incorporate a wide range of ancillary model features such as geographical information or trait data which are essential for some applications.
However, Bayesian tree inference with current implementations is a computationally intensive task, often requiring days or weeks of CPU time to analyze modest datasets with 100 or so sequences.

New developments in DNA and RNA sequencing technology have led to sustained growth in sequence datasets.
This advanced technology has enabled real time outbreak surveillance efforts, such as ongoing Zika, Ebola, and foodborne disease sequencing projects, which make pathogen sequence data available as an epidemic unfolds \citep{Gardy2015-rb,Quick2016-fz}.
In general these new pathogen sequences arrive one at a time (or in small batches) into a background of existing sequences.
Most phylogenetic inferences, however, are performed ``from scratch'' even when an inference has already been made on the previously available sequences.
Thus projects such as \texttt{nextflu.org} \citep{Neher2015-jr} incorporate new sequences into trees as they become available, but do so by recalculating the phylogeny from scratch at each update using a fast approximation to maximum likelihood inference, rather than a Bayesian method.

Modern researchers using phylogenetics are in the situation of having previous inferences, having new sequences, and yet having no principled method to incorporate those new sequences into existing inferences.
Existing methods either treat a previous point estimate as an established fact and directly insert a new sequence into a phylogeny \citep{Matsen2010-ze,Berger2011-zy}, or use such a tree as a starting point for a new maximum-likelihood search \citep{Izquierdo-Carrasco2014-hu}.
There is currently no method to update posterior distributions on phylogenetic trees with additional sequences.

\begin{figure}
\includegraphics[width=0.95\textwidth]{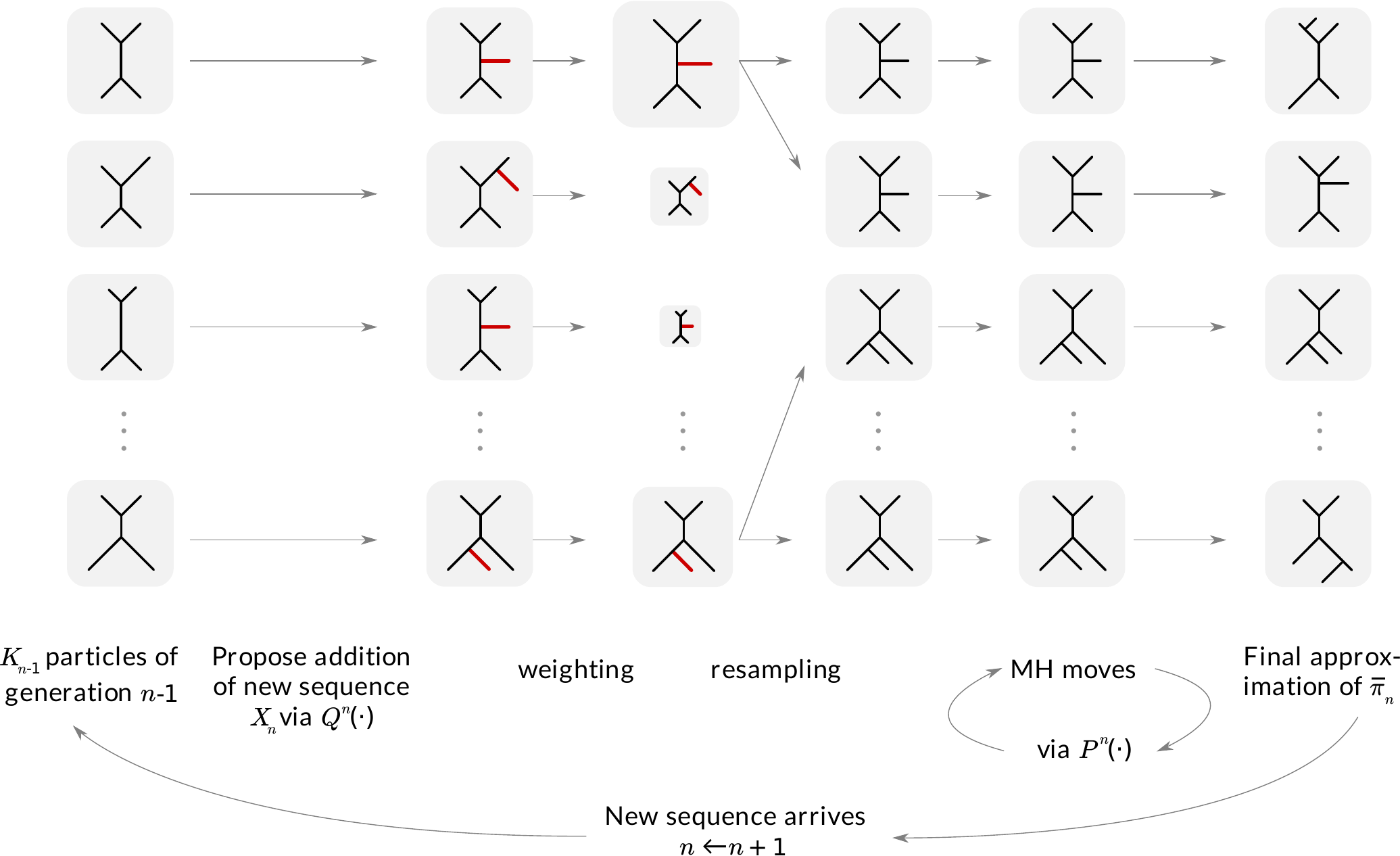}
\caption{An overview of the Online Phylogenetic Sequential Monte Carlo algorithm.}
\label{fig:overview}
\end{figure}

In this paper we develop the theoretical foundations for an online Bayesian method for phylogenetic inference based on Sequential and Markov Chain Monte Carlo.
Unlike previous applications of Sequential Monte Carlo (SMC) to phylogenetics \citep{bouchard2012phylogenetic,Bouchard-Cote2014-fw,Wang2015-nm}, we develop and analyze algorithms that can update a posterior distribution as new sequence data becomes available.
We first show a consistency result, demonstrating that the method samples from the correct distribution in the limit of a large number of particles in the SMC.
Next we derive the first reported set of bounds on how phylogenetic likelihood surfaces change when new sequences are added.
These bounds enable us to characterize the theoretical performance of sampling algorithms by developing a lower bound on the effective sample size (ESS) for a given number of particles.
Surprisingly, this result holds even though the dimensions of the phylogenetic model grow with each new added sequence.

\section{Mathematical setting}

\subsection{Background and notation}
Throughout this paper, a \emph{phylogenetic tree} $(\tau, l)$ is an unrooted tree $\tau$ with leaves labeled by a set of taxon names (e.g.\ species names), such that each edge $e$ is associated with a non-negative number $l_e$.
For each phylogenetic tree $(\tau,l)$, we will refer to $\tau$ as its \emph{tree topology} and to $l$ as the vector of \emph{branch lengths}.
We denote by $E(\tau)$ the set of all edges in trees with topology $\tau$; any edge adjacent to a leaf is called a \emph{pendant edge}, and any other edge is called an \emph{internal edge}.

We will employ the standard likelihood-based framework for statistical phylogenetics on discrete characters under the common assumption that alignment sites are IID \citep{felsenstein2004inferring}, which we now review briefly.
Let $\Omega$ denote the set of character states and let $r= |\Omega|$.
For DNA $\Omega=\{A,C,G,T\}$ and $r=4$.
We assume that the mutation events occur according to a continuous time Markov chain on states $\Omega$ with instantaneous rate matrix $\Xi$ and stationary distribution $\omega$.
This rate matrix $\Xi$ and the branch length $l_e$ on the edge $e$ define the transition matrix $G^e=e^{\Xi l_e}$ on edge $e$, where $G^e_{ij}(l_e)$ denotes the probability of mutating from state $i$ to state $j$ across the edge $e$ (with length $l_e$).

In an online setting, the taxa $\{X_1,X_2,\ldots,X_N\}$ and their corresponding observed sequences $\{\psi_1, \psi_2,\ldots,\psi_N\}$, each of length $S$, arrive in a specific order, where $N$ is a finite but large number.
For all $n \le N$, we consider the set of all phylogenetic trees that have $\{X_1,X_2,\ldots,X_n\}$ as their set of taxa and seek to sample from a sequence of probability distributions $\bar \pi_n$ of increasing dimension corresponding to phylogenetic likelihood functions \citep{felsenstein2004inferring}.

For a fixed phylogenetic tree $(\tau,l)$, the phylogenetic likelihood is defined as follows and will be denoted by $L(\tau, l)$.
Given the set of observations $\psi(n) = (\psi_1, \psi_2,\ldots,\psi_S) \in \Omega^{n \times S}$ of length $S$ up to time $n$, the likelihood of observing $\psi(n)$ given the tree has the form
\[
L_n(\tau, l) = \prod_{u=1}^S{\sum_{a^u}{\omega(a^u_{\rho})\prod_{(i,j)\in E(\tau)}{G^{ij}_{a^u_ia^u_j}( l_{(i,j)})}}}
\]
where $a^u$ ranges over all extensions of $\psi$ to the internal nodes of the tree, $a^u_i$ denotes the assigned state of node $i$ by $a^u$, $\rho$ denotes the root of the tree.
Although we designate a root for notational convenience, the methods and results we discuss apply equally to unrooted trees.

Given a proper prior distribution with density $\pi_0$ imposed on branch lengths and tree topologies, the target posterior distributions can be computed as $\bar \pi_n(\tau, l) \sim L_n(\tau, l) \pi_0(\tau, l)$.
We will also denote by $\hat \pi_n(\tau, l)$ the un-normalized measure $ L_n(\tau, l) \pi_0(\tau, l).$

Throughout the paper,  we assume that the phylogenetic trees of interest all have non-negative branch lengths bounded from above by $b>0$ and denote by $\mathcal{T}_n$ the set of all such trees.
To enable integration on tree spaces and define $\bar \pi_n$, we consider the natural probability measure on $\mathcal{T}_n$: the set $\mathcal{T}_n$ is viewed as the product space of the space of all possible tree topologies (with uniform measure) and the space of all branch lengths $[0,b]^{2n-3}$ (with Lebesgue measure).
These measures can be written as
\[
d \mu_n(\tau,l)= \frac{1}{V_n} d\nu_n(l)= \left(\frac{1}{V_n} d\tau \right)\left(\prod_{e \in E({\tau})}{d\nu(l_e)}\right),
\]
where $V_n = (2n-3)!!$ is the number of different topologies of $\mathcal{T}_n$, $l_e$ is the length of edge $e$, $d\tau$ is the counting measure on the set of all topologies on $\mathcal{T}_n$,  and $d\nu$ is the Lebesgue measure on $\mathbb{R}^+$.

\subsection{Sequential Monte Carlo}

SMC methods are designed to approximate a sequence of probability distributions changing through time.
These probability distributions may be of increasing dimension or complexity.
They track the sequence of probability distributions of interest by producing a discrete representation of the distribution $\bar \pi_n$ at each iteration $n$ through a random collection of weighted particles.
After each generation, new sequences arrive and the collection of particles is updated to represent the next target distribution.
While the details of the algorithms might vary, the main idea of SMC interspersed with MCMC sampling can be described as follows.

At the beginning of each iteration $n$, a list of $K_n$ particles $p^{n}_{1},\ldots,p^{n}_{K_n}$ are maintained along with a positive weight $w^n_i$ associated with each particle $p^{n}_{i}$.
These weighted particles form an un-normalized measure and a corresponding normalized empirical measure
\[
\hat \pi_{n, K_n} =  \sum_{i=1}^{K_{n+1}}{ w^{n}_i\delta_{p^{n}_i}(\cdot)}
\text{ \ and \ }
\bar \pi_{n, K_n} = K_n^{-1} \hat \pi_{n, K_n}
\]
such that $\bar \pi_{n, K_n}$ approximates $\bar \pi_n$.
A new list of $K_{n+1}$ particles is then created in three steps: selection, Markov transition and mutation.

The aim of the selection step is to obtain an unweighted empirical distribution of the weighted measure $\hat{\pi}_{n, K_n}$ by discarding samples with small weights and allowing samples with large weights to reproduce.
Formally, after selection we obtain the unweighted measure
\[
\hat{\alpha}_{n,K_{n+1}}= \sum_{i=1}^{K_n}{ K_{n+1, i} \, \delta_{p^{n}_i}(\cdot)}
\]
where $K_{n+1,i}$ is the multiplicity of particle $p^{n}_i$, sampled from a multinomial distribution parameterized by the weights $w^{n}_i$.
We denote the particles obtained after this step by $s^{n}_{i}$.

The scheme employed in the selection step introduces some Monte Carlo error.
Moreover, when the distribution of the weights from the previous generation is skewed, the particles having high importance weights might be over-sampled.
This results in a depletion of samples (or \emph{path degeneracy}): after some generations, numerous particles are in fact sharing the same ancestor.
A Markov transition step can be employed to alleviate this sampling bias, during which MCMC steps are run separately on each particle $s^{n}_{i}$ for a certain amount of time to obtain a new independent sample $m^{n}_{i}$ with (unweighted) measure denoted $\hat \beta_{n,K_{n+1}}$.

Finally, in the mutation step, new particles $t^{n+1}_{1},\ldots,t^{n+1}_{K_{n+1}}$ are created from a proposal distribution $Q^n$ and are weighted by an appropriate weight function $h$.
If we assume further that for each state $t$, there exists a unique state $s$, denoted by $\varrho(t)$, such that $Q^n(\varrho(t),t)>0$, then $h$ can be chosen as
\begin{equation}
h(t) = \frac{{\hat\pi}_{n+1}(t)}{{\hat\pi}_n(\varrho(t)) ~ Q^n(\varrho(t),t)}.
\label{eq1}
\end{equation}
The process is then iterated until $n=N$.

For convenience, we will denote the unnormalized empirical measures of the particles right after step $n$ by $\hat \alpha_{n,K_{n+1}}$, $\hat \beta_{n,K_{n+1}}$ and $\hat \lambda_{n,K_{n+1}}$, respectively.
%
% Is that the unweighted or unnormalized measure?
% AD: The above sentence is no longer clear
%V  These empirical measures are all un-normalized  (The corresponding normalized are defined just for theoretical analysis, we will not actually use those).
% Among those empirical measures, some are weighted and some are unweighted (all particles have same weight). To be precise, the particles after selection and before mutation are unweighted.
Similarly, the corresponding normalized distributions will be denoted by $\bar \alpha_{n, K_{n+1}}$, $\bar \beta_{n, K_{n+1}}$ and $\bar \lambda_{n, K_{n+1}}$.

\section{Online phylogenetic inference via Sequential Monte Carlo}

Here we develop Online Phylogenetic sequential Monte Carlo (OPSMC) methods that continually update phylogenetic posteriors as new molecular sequences are added.
In contrast to the traditional setting of SMC, for OPSMC when the number of leaves $n$ of the particles increases, not only does the local dimension of the space $\mathcal{T}_n$ increase linearly, the number of different topologies in $\mathcal{T}_n$ also increases super-exponentially in $n$.
Careful constructions of the proposal distribution $Q^n$, which will build $n+1$-taxon trees out of $n$-taxon trees, and the Markov transition kernel $P^n$ are essential to cope with this increasing complexity.

Given two trees $r$ and $r'$ in the tree space $\mathcal{T}=\bigcup{~\mathcal{T}_n}$, we say that $r'$ \emph{covers} $r$ if there exists $n$ such that $r \in \mathcal{T}_n$, $r' \in \mathcal{T}_{n+1}$, and $r$ can be obtained from $r'$ by removing the taxon $X_{n+1}$ and its corresponding edge.
This definition is analogous to the covering definition of \cite{Wang2015-nm}, although is distinct in the setting of online inference.
The proposal distributions $Q^n$ will be designed in such a way that the following criterion holds.
\begin{Criterion}
At every step of the OPSMC sampling process, the proposal density $Q^n$ satisfies $Q^n(r,r')>0$ if and only if $r'$ covers $r$.
\label{hasse}
\end{Criterion}

Under this criterion, for every tree $t \in \mathcal{T}_{n+1}$, there exists a unique tree $\varrho(t)$ in $\mathcal{T}_{n}$ such that $Q^n(\varrho(t),t)>0$ and thus a weight function of the form $\eqref{eq1}$ can be used.

To obtain an $(n+1)$-taxon tree from an $n$-taxon tree, a proposal strategy $Q^n$ must specify:
\begin{enumerate}
\item an edge $e$ to which the new pendant edge is added,
\item the position $x$ on that edge to attach the new pendant edge, and
\item the length $y$ of the pendant edge.
\end{enumerate}
The position $x$ on an edge of a tree will be specified by its \emph{distal} length, which is the distance from the attachment location to the end of the edge that is farthest away from the root of the tree.
Different ways of choosing $(e,x,y)$ lead to different sampling strategies and performances.
Throughout the paper, we will investigate two different classes of sampling schemes: length-based proposals and likelihood-based proposals.

\subsection{Length-based proposals}
For length-based proposals:
\begin{enumerate}
\item the edge $e$ is chosen from a multinomial distribution weighted by length of the edges,
\item the distal position $x$ is selected from a distribution $P^e_X(x)$ across the edge length,
\item the pendant length $y$ is sampled from a distribution $P_Y(y)$ with support contained in $[0, b]$.
\end{enumerate}
For example, if these $P$ distributions are uniform, we obtain a uniform (with respect to Lebesgue measure) prior on attachment locations across the tree.
We assume that

\begin{Assumption}
The densities $p^e_X$ of the distal position on edge $e$ and $p_Y$ of the pendant edge lengths are absolutely continuous with respect to the Lebesgue measure on $[0,l_e]$ and $[0,b]$, respectively.
Moreover,
\[
\frac{1}{l_e^{2}}\int_{0}^{l_e}{\frac{1}{p^e_X(x)}~d\nu(x)} \ \le C \h \text{and} \h \int_0^{\infty}{ \frac{1}{p_Y(y)}~d\nu(y)} \ < \infty.
\]
where $l_e$ denotes the length of edge $e$ and $C$ is independent of $l_e$.
\label{distal}
\end{Assumption}

We note that, for any density function $\psi$ on $[0,1]$ such that $1/\psi$ is integrable, the family of proposals $\psi_{l}(x) = \frac{1}{l} \psi(\frac{x}{l})$ satisfies Assumption~$\ref{distal}$.
The densities $p^e_X$ and $p_Y$  are assumed to be absolutely continuous to ensure Criterion~$\ref{hasse}$ holds.

As we will discuss in later sections, to make sure that the proposals can capture the posterior distributions $\bar \pi_n$ efficiently, some regularity conditions on $\bar \pi_n$ are also necessary.
These conditions are formalized in terms of a lower bound on the posterior expectation of $\zeta(s)$, the average branch length of $s$ for a given tree $s \in \mathcal{T}_n$.

\begin{Assumption}[Assumption on the average branch length]
There exist positive constants $c$ (independent of $n$) such that for each $n$
\[
c\le \int_{\mathcal{T}_n}{\bar \pi_n(r)\zeta(r)~dr}
\]
where $\zeta(r)$ denotes the average of branch lengths of the tree $r$.
\label{branchlength}
\end{Assumption}

\subsection{Likelihood-based proposals}

In the likelihood-based approach, the edge $e$ (from the tree $r$) is chosen from a multinomial distribution weighted by a likelihood-based utility function $f(s, e)$.
Similarly, the distributions $P^e_X(x)$ and $P_Y(y)$ might also be guided by information about the likelihood function.
Likelihood-based proposals are capable of capturing the posterior distribution more efficiently, but with an additional cost for computing the likelihoods.

We define the average likelihood utility function
\[
\mathcal{G}_n(r,e)= \int_{x,y}{\hat \pi_{n+1}(T(r,e, x,y))~dx~dy}
\]
and use it as the prototype for likelihood-based utility functions.
The likelihood-based utility function $f(r, e)$ is assumed to satisfy the following assumption.

\begin{Assumption}
There exist $c_1, c_2>0$ such that $c_1 \mathcal{G}_n(r,e) \le f_n(r,e) \le c_2 \mathcal{G}_n(r,e)$ for all $r, e$.
\label{utility}
\end{Assumption}
The following lemma (proven in the Appendix) establishes that the maximum likelihood utility function also satisfies Assumption $\ref{utility}$.

\begin{Lemma}
Let $f_n(r,e) = b \, l_e\sup_{x,y}{~\hat \pi_{n+1}(T(r,e, x,y))}$, there exists $c_3>0$ independent of $n$ such that $\mathcal{G}_n(r,e) \le f_n(r,e) \le c_3 \mathcal{G}_n(r,e)$ for all $s, e$.
\label{maximum}
\end{Lemma}

As for the length-based proposal, we assume the following conditions on the distal position and pendant edge length proposals for the likelihood-based approach.
\begin{Assumption}
The densities $p^e_X$ and $p_Y$ are absolutely continuous with respect to the Lebesgue measure on $[0,l_e]$ and $[0,b]$, respectively.
Moreover, there exists $a_0$ independent of $n$ such that
\[
\sup_{x,y}{\frac{1}{p^e_X(x)} ~ \frac{1}{p_Y(y)}} \le a_0.
\]
\label{uniform}
\end{Assumption}

\subsection{Markov transition kernels}
% trimmable
%V Removed repetitive texts
Besides the SMC proposal strategy $Q^n$, it is also important to choose an appropriate Markov transition kernel $P^n$ to have an effective OPSMC algorithm.
It is worth noting that the problem of sample depletion is even more severe for OPSMC, since after each generation, the sampling space actually expands in dimensionality and complexity.
To alleviate this sampling bias, MCMC steps are run separately on each particle  for a certain amount of time to obtain new independent samples.
We require the following criterion, which is as expected for any Markov transition kernel used in standard MCMC.

\begin{Criterion}
At every step of the OPSMC sampling process, the Markov transition kernel $P^n$ has $\bar{\pi}_{n}$ as its invariant measure.
\label{kernel}
\end{Criterion}

As we will see later in the proof of consistency of OPSMC, Criterion~$\ref{kernel}$ is the only assumption to be imposed on the Markov transition kernel.
This leaves us with a great degree of freedom to improve the efficiency of the sampling algorithm without damaging its theoretical properties.
For example, this allows us to use global information provided by the population of particles, such as effective sample size \citep{beskos2014stability}, to guide the proposal, or to define a transition kernel on the whole set (or some subset) of particles \citep{andrieu2001sequential}.
In the context of phylogenetics, we can design a sampler that recognizes subtrees that have been insufficiently sampled, and samples more particles to improve the effective sample size within such regions.
% ^--- do you mean runs more particles or runs extra MCMC steps on particles with those subtrees?
% what about approaches that use conditional clade probabilities? (though I guess they have problems in some cases, but should be acceptable, no?)
Similarly, one can use samplers that rearrange the tree structure in the neighborhood of newly added pendant edges.

\section{Consistency of online phylogenetic SMC}

In this section, we establish the consistency of OPSMC in the limit of a large number of particles by induction on the number of taxa $n$; that is, for every $n<N$, assuming that $\bar\pi_{n,K_n} \to \bar\pi_{n} $, we will prove that $\bar\pi_{n+1,K_{n+1}} \to \bar\pi_{n+1}$.
We note that although the measures mentioned above are indexed by $K_n$, they implicitly depend on the number of particles from the previous generations.
Thus, the convergence should be interpreted in the sense of when the number of particles of all generations approaches infinity.

The mode of convergence used in this section is ``weak convergence", in which we say $\mu_K \to \mu$ if for every appropriate test function $\phi$ we have $\lim_{K \to \infty}{\int{\phi(t)d \mu_K(t)}} = \int{\phi(t)d \mu(t)}$.
We will use $\mu(\phi)$ to denote $\int{\phi(t)d \mu(t)}$ for any measures $\mu$ and test functions $\phi$.

For convenience, let $L$ and $K$ be the number of particles at $n^\text{th}$ and $(n+1)^\text{st}$ generation, respectively.
Recall that the normalized distributions after the substeps of OPSMC are denoted by $\bar\alpha_{n,K}$, $\bar \beta_{n,K}$ and $\bar \lambda_{n,K}$, we have the following lemma, proven in the Appendix.
\begin{Lemma}
Assume that Criteria~$\ref{hasse}$ and $\ref{kernel}$ are satisfied.
If we define
\[
\bar{\lambda}_{n}(t):=\bar{\pi}_{n}(\varrho(t))Q^n(\varrho(t),t) \h \text{and} \h
h(t) := \frac{{\hat\pi}_{n+1}(t)}{{\hat\pi}_n(\varrho(t)) ~ Q^n(\varrho(t),t)}
\]
then the following statements hold.
\begin{enumerate}
\item If $\bar\pi_{n,L}$ converges to $\bar\pi_{n} $, then $\bar \alpha_{n,K}$ converges to $\bar\pi_{n} $.
\item If $\bar \alpha_{n,K}$ converges to $\bar\pi_{n}$, then $\bar \beta_{n,K}$ converges to $\bar\pi_{n}$.
\item If $\bar \beta_{n,K}$ converges to $\bar\pi_{n}$, then $\bar\lambda_{n,K}$ converges to $\bar\lambda_{n}$.
\item $h(t)\bar\lambda_{n}(t)$ is proportional to $\bar\pi_{n+1}(t)$.
\item If $\bar\lambda_{n,K}$ converges to $\bar\lambda_{n}$, then $\bar\pi_{n+1,K}$ converges to $\bar\pi_{n+1} $.
\end{enumerate}
\label{induction}
\end{Lemma}

We note that when $n=0$, the set of all rooted trees with no taxa consists of a single tree $\rho$.
Thus, if we use this single tree as the ensemble of particles at $n=0$, then $\bar \pi_{0, K_0}$ is precisely $\bar \pi_0$.
Alternatively, we can start with $n=1$ and use some ergodic MCMC methods to create an ensemble of particles with stationary distribution $\bar \pi_1$.
In either case, an induction argument with Lemma $\ref{induction}$ gives the main theorem:

\begin{Theorem}[Consistency]
If Criteria~$\ref{hasse}$ and $\ref{kernel}$ are satisfied and the sampler starts at $n=0$ by a list consisting of a single rooted tree with no taxa, or at $n=1$ with an ensemble of particles created by an ergodic MCMC method with stationary distribution $\bar \pi_1$, then
\[
\bar \pi_{n,K_n}(\phi) \to \bar \pi_n(\phi) \h  \text{as} \h K_1, K_2, \ldots K_n \to \infty
\]
for every integrable test function $\phi: \mathcal{T}_n \to \mathbb{R}$ and $n \le N$.
\label{consistency}
\end{Theorem}

\section{Characterizing changes in the likelihood landscapes when new sequences arrive}

Although the consistency of OPSMC is guaranteed and informative OPSMC samplers can be developed by changing the Markov transition kernels, its applicability is constrained by an implicit assumption: the distance between target distributions of consecutive generations are not too large.
Since SMC methods are built upon the idea of recycling particles from one generation to explore the target distribution of the next generation, it is obvious that one would never be able to design an efficient SMC sampler if $\bar \pi_n$ and $\bar \pi_{n+1}$ are effectively orthogonal.

While a condition on minor changes in the target distributions may be easy to verify in some applications, it is not straightforward in the context of phylogenetic inference.
A similar question on how the ``optimal" trees (under some appropriate measure of optimality) change has been studied extensively in the field, with negative results for almost all regular measures of optimality \citep{heath2008taxon, cueto2011polyhedral}.
To the best of our knowledge, no previous work has been done to investigate how phylogenetic likelihood landscapes change when new sequences arrive.

In this section, we will establish that under some minor regularity conditions on the distribution described in the previous sections, the relative changes between target distributions from consecutive generations are uniformly bounded.
This result enables us to provide a lower bound on the effective sample size of OPSMC algorithms in the next section.

We denote by $T(r,e,x,y)$ the tree obtained by adding an edge of length $y$ to edge $e$ of the tree $r$ at distal position $x$.
Thus, any tree $t$ can be represented by $t=(\varrho(t),e(t),x,y)$, where $e(t)$ is the edge on which the pendant edge containing the most recent taxon is attached.

\begin{Lemma}[Change of variables]
The map $(r,e,x,y) \to T(r,e,x,y)$ is bijective.
 Moreover,
\[
d\mu_{n+1}(t) = \frac{V_n}{V_{n+1}} dx~dy~ de~d\mu_n(r)
\]
where $de$ is the counting measure on the set of edges of an $n$-tree, and again $V_n = (2n-3)!!$.
\label{lem:changeVariables}
\end{Lemma}

This result allows us to derive the following Lemma (detailed proof is provided in the Appendix).

\begin{Lemma}
Consider an arbitrary tree $t \in \mathcal{T}_{n+1}$ obtained from the parent tree $\varrho(t)$ by choosing edge $e$, distal position $x$ and pendant length $y$.
Denote
\[
M(y) = \max_{ij}{G_{ij}(y)}, \h m(y) = \min_{ij}{G_{ij}(y)} \h \text{and} \h \mathcal{Z}_{n}=\int_{s \in \mathcal{T}_{n}}{{\bar\pi_{n}(s)}\zeta(s)~ds}.
\]
We have
\[
\frac{\bar \pi_{n+1}(t)}{\bar \pi_n (\varrho(t))} \le \frac{1}{\mathcal{Z}_n} \frac{M(y)^S}{\int_0^{b}{m(y)^S \, dy}}, \h \forall t \in \mathcal{T}_{n+1}.
\]
\label{subtree}
\end{Lemma}

\begin{proof}[Sketch of proof]
By using the one-dimensional formulation of the phylogenetic likelihood function derived in \citep{dinh2015shape}, we can prove that
\begin{equation}
\frac{\hat \pi_{n+1}(t)}{\hat \pi_n (\varrho(t))}=\frac{L_{n+1}(t)}{L_n (\varrho(t))} \le M(y)^S, \h \forall t \in \mathcal{T}_{n+1}.
\label{util1}
\end{equation}
Similarly, we have $\hat \pi_{n+1}(t)/\hat \pi_n (\varrho(t)) \ge m(y)^S$ for all $t \in \mathcal{T}_{n+1}$.

Recall that $\zeta(r)$ is the average branch length of $r$.
Using the fact that for a fixed tree $r$, $\int_0^{l_e}{dx = l_e}$ and $\sum_{e}{l_e} = (2n-3) \zeta(r)$, we have
\begin{align*}
\|\hat \pi_{n+1}\| =& \int_{t \in \mathcal{T}_{n+1}}{\hat \pi_{n+1}(t) \, dt}
\ge \int_{t \in \mathcal{T}_{n+1}}{m(y)^S \hat \pi_{n}(\varrho(t)) \, dt}\\
=& \int_{r,e,x,y}{m(y)^S \hat\pi_n(r) \, dx \, dy} ~ \frac{V_n}{V_{n+1}}de \, dr\\
=& \frac{(2n-3)V_n}{V_{n+1}} \left(\int_0^{b}{m(y)^S \, dy}\right)\int_{r \in \mathcal{T}_{n}}{{\hat \pi_{n}(r)}\zeta(r) \, dr}
\end{align*}
which implies
\[
\frac{\bar \pi_{n+1}(t)}{\bar \pi_n (\varrho(t))} = \frac{\hat \pi_{n+1}(t)}{\hat \pi_n (\varrho(t))}  \frac{\|\hat \pi_{n}\|}{\|\hat \pi_{n+1}\|} \le \frac{1}{\mathcal{Z}_n} \frac{M(y)^S}{\int_0^{b}{m(y)^S~dy}}, \h \forall t \in \mathcal{T}_{n+1}.
\]
\end{proof}

%AD If possible, I suggest adding an intuitive explanation for this result
% for the benefit of readers without the patience (or aptitude) for the formal explanation

\section{Effective sample sizes of online phylogenetic SMC}

In this section, we are interested in the asymptotic behavior of OPSMC in the limit of large $K_n$, i.e. when the number of particles of the sampler approaches infinity.
This asymptotic behavior is illustrated via estimates of the effective sample size of the sampler with large numbers of particles.
We note that although there are several studies on the stability of SMC as the time step grows, most of them focus on cases where the sequence of target distributions have a common state space of fixed dimension \citep{del1998uniform, douc2008limit, kunsch2005recursive, oudjane2005stability, del2009tree, beskos2014stability}.
In general, establishing stability bounds for SMC requires imposing some conditions on the effect of data at any step $k$ to the target distribution at step $n \ge k$ \citep{crisan2002survey, chopin2004central, doucet2009tutorial}.
Lemma $\ref{subtree}$ helps validate a condition of this type.

The effective sample size \citep{beskos2014stability} of the particles at step $n+1$ is computed as
\[
\text{ESS}_{n+1}= \frac{\left(\sum_{i=1}^K{w^{n+1}_i}\right)^2}{\sum_{i=1}^K{(w^{n+1}_i)^2}}.
\]
The following result, proven in the Appendix, enables us to estimate the asymptotic behavior of the sample's \text{ESS} in various settings.

\begin{Theorem}
In the limit as the number of particles approaches infinity, we have
\[
\lim_{K \to \infty}{\frac{K_{n+1}}{\text{ESS}_{n+1}}} = \int_{t \in \mathcal{T}_{n+1}}{\frac{\bar \pi_{n+1}^2(t)}{\bar \pi_{n}(\varrho(t))~Q^n(\varrho(t),t)}~dt}.
\]
\label{thm:lengthESSLimit}
\end{Theorem}

This asymptotic estimate and the results on likelihood landscapes from the previous section allow us to prove the following Theorem.

\begin{Theorem}[Effective sample size of OPSMC for likelihood-based proposals]
If Assumptions~$\ref{branchlength}$, $\ref{uniform}$ and $\ref{utility}$ hold, then there exists $\alpha>0$ independent of $n$ such that $\text{ESS}_{n} \ge \alpha K_{n}$.
That is, the effective sample size of an OPSMC with likelihood-based proposals are bounded below by a constant multiple of the number of particles.
Moreover, if Assumption $\ref{branchlength}$ does not hold, the effective sample size of OPSMC algorithms decays at most linearly as the dimension increases.
\label{likelihoodbased}
\end{Theorem}

\begin{proof}[Proof of Theorem $\ref{likelihoodbased}$]

Define $f_n(r)= \sum_{e}{f_n(r,e)}$, we have
\[
\int_{r}{f_n(r)~ds} \le c_2 \int_{r}{ \sum_{e}{\int_{x,y}{\hat \pi_{n+1}(T(r,e, x,y))~dx~dy}}~dr}
=  c_2 (2n-3) \|\hat \pi_{n+1}\|.
\]
Since edge $e$ is chosen from a multinomial distribution weighted by $f_n(r,e)$, given any tree $t \in \mathcal{T}_{n+1}$ obtained from the parent tree $\varrho(t)$, chosen edge $e(t)$, distal position $x$ and pendant length $y$,
\[
Q^n(\varrho(t),t) = \frac{V_{n+1}}{V_n} \frac{f(\varrho(t), e(t))}{f(\varrho(t))}~ p_X(x) ~ p_Y(y).
\]
By Lemma $\ref{subtree}$ and the fact that $M(y) \le 1$, we have
\[
\frac{\bar \pi_{n+1}(t)}{\bar \pi_n (\varrho(t))} \le \frac{1}{\mathcal{Z}_n} \frac{M(y)^S}{\int_0^{b}{m(y)^S \, dy}} \le \frac{1}{u_1 \mathcal{Z}_n},
\]
where $u_1 = \int_0^{b}{m(y)^S~dy}$ and $\mathcal{Z}_n$ are defined as in the proof of Lemma $\ref{lengthbased}$.
Using Assumptions~$\ref{branchlength}$ and $\ref{uniform}$, $\eqref{eq-1}$, Lemma~\ref{lem:changeVariables} and similar arguments as in the previous proof, we have
\begin{align*}
&\int_{t \in \mathcal{T}_{n+1}}{\frac{\bar \pi_{n+1}^2(t)}{\bar \pi_{n}(\varrho(t))~Q^n(\varrho(t),t)}~dt} \\
&\le a_0 \frac{1}{u_1 \mathcal{Z}_n} \left(\frac{V_n}{V_{n+1}}\right)^2   \int_{r,e}{ \frac{f_n(r)}{f_n(r,e)} \int_{x,y}{\bar \pi_{n+1}(T(r,e,x,y))~dx ~dy}~dr ~de}\\
%Note for EM: \frac{1}{\|\hat \pi_{n+1}\|} comes from going from pi bar to G_n (which is in terms of pi hat.
&\le \frac{a_0}{c_1} \frac{1}{u_1 \mathcal{Z}_n} \frac{1}{\|\hat \pi_{n+1}\|} \left(\frac{V_n}{V_{n+1}}\right)^2  \left(\int_{r}{ f_n(r)ds}\right) \left( \int_{e}{de}\right)\\
&\le (2n-3)^2 a_0 \frac{c_2}{c_1} \frac{1}{u_1 \mathcal{Z}_n}  \left(\frac{V_n}{V_{n+1}}\right)^2 = a_0 \frac{c_2}{c_1} \frac{1}{u_1 \mathcal{Z}_n}.
\end{align*}
Thus by Theorem~\ref{thm:lengthESSLimit} there exists $\alpha>0$ independent of $K_n$ and $n$ such that $\text{ESS}_{n} \ge \alpha K_n$.
We also note that without the assumption on average branch lengths, a crude estimate gives $\mathcal{Z}_n \ge \mathcal{Z}_1/n$, which leads to a linear decay in the upper bound on the ESS.
\end{proof}

We also have similar estimates for length-based proposals (see Appendix for proof):

\begin{Theorem}[Effective sample size of OPSMC for length-based proposals]
If Assumptions~$\ref{distal}$ and $\ref{branchlength}$ hold, then the effective sample size of OPSMC with length-based proposals are bounded below by a constant multiple of the number of particles.
Moreover, if Assumption $\ref{branchlength}$ does not hold, the effective sample size of OPSMC algorithms decays at most quadratically as the dimension increases.
\label{lengthbased}
\end{Theorem}

In summary, we are able to prove that in many settings, the effective sample size of OPSMC is bounded from below.
These results are interesting, since in the general case it is known that SMC-type algorithms may suffer from the curse-of-dimensionality:  when the dimension of the problem increases, the number of the particles must increase exponentially to maintain a constant effective sample size \citep{chopin2004central, bengtsson2008curse, bickel2008sharp, snyder2008obstacles}.

\section{Discussion}
In this paper, we establish foundations for Online Phylogenetic Sequential Monte Carlo (OPSMC), including essential theoretical convergence results.
We prove that under some mild regularity conditions and with carefully constructed proposals, the OPSMC sampling algorithm is consistent.
%AD is it worth saying something about how avg. branch length is expected to behave in practice? even just speculative without data? we should have a lot of data from which conclusions could be drawn... my intuition is that in many datasets total branch length would grow quickly at first then very slowly, while average branch length might decrease at first then plateau, unless an adversary started to add increasingly divergent sequences.
%EM Personally I don't think that it's necessary. At least to me it seems like a necessary technical condition with a clear interpretation.
This includes relaxing the condition used in \cite{bouchard2012phylogenetic}, in which the authors assume that the weight of the particles are bounded from above.
We then investigate two different classes of sampling schemes for online phylogenetic inference: length-based proposals and likelihood-based proposals.
In both cases, we show the effective sample size to be bounded below by a multiple of the number of particles.

The consistency and convergence results in this paper apply to a variety of sampling strategies.
One possibility would be for an algorithm to use a large number of particles, directly using the SMC machinery to approximate the posterior.
Alternatively, the SMC part of the sampler could be quite limited, resulting in an algorithm which combines many independent parallel MCMC runs in a principled way.
As described above, the SMC portion of the algorithm enables MCMC transition kernels that would normally be disallowed by the requirement of preserving detailed balance.
For example, one could use a kernel that focuses effort around the part of the tree which has recently been disturbed by adding a new sequence.

In the future we will develop efficient and practical implementations of these ideas.
Many challenges remain.
For example, the exclusive focus of this paper has been on the tree structure, consisting of topology and branch lengths.
However, Bayesian phylogenetics algorithms typically co-estimate mutation model parameters along with tree structures.
Although proposals for other model parameters can be obtained by particle MCMC \citep{Andrieu2010-us}, we have not attempted to incorporate it into the current SMC framework.
In addition, we note that the input for this type of phylogenetics algorithm consists of a \emph{multiple sequence alignment} (MSA) of many sequences, rather than just individual sequences themselves.
This raises the question of how to maintain an up-to-date MSA.
Programs exist to add sequences into existing MSAs \citep{Caporaso2010-dm,Katoh2013-dl}, although from a statistical perspective, it could be preferable to jointly estimate a sequence alignment and tree posterior \citep{Suchard2006-lx}.
It is an open question how that could be done in an online fashion, although in principle it could be facilitated by some modifications to the sequence addition proposals described here.

\newpage

\bibliographystyle{plainnat}
\bibliography{sequential}

\newpage

\section{Appendix}

\begin{proof}[Proof of Lemma $\ref{maximum}$]
The lower bound is straightforward.
For the upper bound, consider $(x,y) \in [0, l_e] \times [0, b]$ and fix $\delta>0$; by the same arguments as in the proof of Lemma $\ref{subtree}$, we have
\[
\hat \pi_{n+1}(T(r,e, x,y)) \ge m(\delta)^S \hat \pi_n(r) \h \forall y \ge \delta.
\]
Thus, if we define
\[
A=\{(x,y) \in [0, l_e] \times [0, b]:  \hat \pi_{n+1}(T(r,e, x,y)) \ge m(\delta)^S \hat \pi_n(r)\},
\]
then we have $|A| \ge (b-\delta) l_e$ and
\begin{align*}
\mathcal{G}_n(r,e) &= \int_{x,y}{\hat \pi_{n+1}(T(r,e, x,y))~dx~dy} \\
& \ge \int_{A}{\hat \pi_{n+1}(T(r,e, x,y))~dx~dy} \ge (b-\delta) l_e ~m(\delta)^S \hat \pi_n(r).
\end{align*}
On the other hand, from Lemma $\ref{subtree}$, we have $f_n(r,e) \le b \, l_e ~\hat \pi_n(r) M(b)^S$.

By choosing $\delta=b/2$, we obtain
\[
f_n(r,e) \le 2 \frac{M(b)^S}{m(b/2)^S} ~ \mathcal{G}_n(r,e)
\]
which completes the proof.
\end{proof}

\begin{proof}[Proof of Lemma $\ref{induction}$]
(1). Assume that $\bar\pi_{n,L}$ converges to $\bar\pi_{n}$.
\[
 |\bar\alpha_{n,K}(\phi)-  \bar\pi_n(\phi)| \le  \left| \frac{1}{K}\sum_{i=1}^{L}{ K_{n+1,i} ~ \phi({p^{n}_i})} - \frac{1}{\|w\|} \sum_{i=1}^{L}{ w_i~ \phi({p^{n}_i})} \right| +  |\bar\pi_{n,L}(\phi)-  \bar\pi_n(\phi)|.
\]
By the strong law of large numbers,
 \[
 \limsup_{K \to \infty} |\bar\alpha_{n,K}(\phi)-  \bar\pi_n(\phi)| \le  |\bar\pi_{n,L}(\phi)-  \bar\pi_n(\phi)|
 \]
This implies that when $K, L \to \infty$, we have $ \bar\alpha_{n,K}(\phi) \to  \bar\pi_n(\phi)$.

(2). The rationale behind the use of MCMC moves is based on the observation that if the unweighted particles are distributed according to $\bar{\pi}_{n}$, then when we apply a Markov transition kernel $P$ of invariant distribution $\bar{\pi}_{n}$ to any particle, the new particles are still distributed according to the posterior distribution of interest.

Formally, if $\bar{\alpha}_{n,K}(\phi) \to \pi_{n}(\phi)$ for every integrable test function $\phi: \mathcal{T}_n \to \mathbb{R}$, by choosing $\phi=P^r(\cdot,A)$ for any measurable set $A \subset \mathcal{T}_{n+1}$, we deduce that
\[
\bar{\beta}_{n+1,K}(A)=\sum_{i=1}^K{P^r(s,A)~ \bar{\alpha}_{n,K}(s^{n}_i)}
\xrightarrow{K \to \infty} \int_{\mathcal{T}_{n}}{P^r(s,A) ~\bar{\pi}_{n}(s)~ds}
 = \bar{\pi}_{n}(A)
\]
since the Markov kernel $P$ is invariant with respect to $\hat \pi_n$.
Therefore, for any  measurable function $\phi : \mathcal{T}_{n} \to \mathbb{R}$, we have that $\bar \beta_{n,K}(\phi)$ converges to $\bar\pi_{n}(\phi)$.

(3). Since $\bar{\beta}_{n,K}(\phi) \to \pi_{n}(\phi)$ for every integrable test function $\phi: \mathcal{T}_n \to \mathbb{R}$, by choosing $\phi=Q^n(\cdot,t)$ for all $t \in \mathcal{T}_{n+1}$, we deduce that
\begin{eqnarray*}
\bar{\lambda}_{n,K}(t)&=&\sum_{i=1}^K{Q^n(m^{n}_i,t) \bar{\beta}_{n,K}(m^{n}_i)}. \\
&\xrightarrow{K \to \infty}& \int_{\mathcal{T}_{n}}{Q^n(m,t) \bar{\pi}_{n}(m)~dm}  = \bar{\pi}_{n}(\varrho(t)) \, Q^n(\varrho(t),t) = \bar{\lambda}_{n}(t) .
\end{eqnarray*}

Moreover, for every measurable set $A \subset \mathcal{T}_{n+1}$ , we can use the same argument to prove that
\begin{eqnarray*}
\bar{\lambda}_{n,K}(A)&=&\sum_{i=1}^K{Q^n(m^{n}_i,A) \bar{\beta}_{n,K}(m^{n}_i)} \\
&\xrightarrow{K \to \infty}& \int_{\mathcal{T}_{n}}{Q^n(m,A) \bar{\pi}_{n}(dm)}  = \int_{A}{\bar{\pi}_{n}(\varrho(t))Q^n(\varrho(t),t) \, dt} = \bar\lambda_{n}(A).
\end{eqnarray*}
Therefore, for any  measurable function $\phi : \mathcal{T}_{n+1} \to \mathbb{R}$, we also have $\bar \lambda_{n,K}(\phi)$ converges to $\bar\lambda_{n}(\phi)$.

(4) We note that
\begin{eqnarray}
h(t)\bar{\lambda}_{n}(t) &=& h(t) \, \bar{\pi}_{n}(\varrho(t)) \, Q^n(\varrho(t),t) \nonumber\\
&=& \frac{{\hat\pi}_{n+1}(t)} {{\hat\pi}_n(\varrho(t)) \, Q^n(\varrho(t),t)} \, \frac{1}{\|\hat\pi_n\|} \, {\hat\pi}_n(\varrho(t)) \, Q^n(\varrho(t),t) \nonumber \\
&=& \frac{1}{\|\hat\pi_n\|} {\hat\pi}_{n+1}(t).
\label{e1}
\end{eqnarray}

(5). Since the proposal $Q^n$ and the Markov kernel $P$ are assumed to be normalized, we have $\|\hat\lambda_{n,K}\|= \|\hat \beta_{n,K}\|=\|\hat \alpha_{n,K}\| = K$.

We have:
%EM very minor annoyance-- the equation number here lines up with the second line of the equation.
%V For some reason, I cannot turn of the indexing of the second line. Will come back to this later.
%V Fixed
\begin{align}
{\frac{1}{K}\|\hat{\pi}_{n+1,K}\|} =& \frac{1}{\|\hat \lambda_{n,K}\|}\sum_{i=1}^{K}{\hat{\pi}_{n+1,K}(t^{n+1}_i)} =  \sum_{i=1}^{K}{h(t^{n+1}_i) \bar{\lambda}_{n,K}(t^{n+1}_i)}  \label{eqess} \\
  \xrightarrow{K \to \infty}& \int_{\mathcal{T}_{n+1}} {{h(t)\bar\lambda}_{n}(t) \, dt} = \frac{1}{\|\hat\pi_n\|} \int_{\mathcal{T}_{n+1}} {{\hat\pi}_{n+1}(t) \, dt} = \frac{\|\hat \pi_{n+1}\|}{\|\hat\pi_n\|}\notag.
\end{align}
By a similar argument, we have
\begin{align*}
\bar{\pi}_{n+1,K}(\phi) &= \frac{\frac{1}{K}\sum_{i=1}^{K}{\phi(t^{n+1}_i)\hat{\pi}_{n+1,K}(t^{n+1}_i)}}{\frac{1}{K}\|\hat{\pi}_{n+1,K}\|} = \frac{\sum_{i=1}^{K}{\phi(t^{n+1}_i) h(t^{n+1}_i) \bar{\lambda}_{n,K}(t^{n+1}_i)}}{\frac{1}{K}\|\hat{\pi}_{n+1,K}\|} \\
&\xrightarrow{K \to \infty} \frac{\|\hat \pi_{n}\|}{\|\hat\pi_{n+1}\|}  \int_{\mathcal{T}_n}{\phi(t) h(t) \bar{\lambda}_{n}(t)  dt }= \frac{\|\hat \pi_{n}\|}{\|\hat\pi_{n+1}\|}  \int_{\mathcal{T}_n}{\phi(t)  \frac{{\hat\pi}_{n+1}(t)}{\|\hat\pi_n\|}dt }= \bar{\pi}_{n+1}(\phi).
\label{eqn_consistency}
\end{align*}
In other words, $\bar\pi_{n+1,K}$ converges to $\bar\pi_{n+1}$.
\end{proof}

\begin{proof}[Proof of Theorem $\ref{thm:lengthESSLimit}$]
By definition, we have
\[
h(t) = \frac{{\hat\pi}_{n+1}(t)}{{\hat\pi}_n(\varrho(t)) ~ Q^n(\varrho(t),t)}, \h \text{and} \h
w^{n+1}_i=\hat{\pi}_{n+1,K_{n+1}}(t^{n+1}_i) = h(t^{n+1}_i) .
\]
Thus,
\begin{eqnarray*}
\bar{\pi}_{n+1,K_{n+1}}(h) &=& \frac{\sum_{i=1}^{K_{n+1}}{h(t^{n+1}_i)\hat{\pi}_{n+1,K}(t^{n+1}_i)}}{\|\hat{\pi}_{n+1,K_{n+1}}\|}\\
&=& \frac{\sum_{i=1}^{K_{n+1}}{(w^{n+1}_i)^2}}{\sum_{i=1}^{K_{n+1}}{w^{n+1}_i}} = \frac{\sum_{i=1}^{K_{n+1}}{w^{n+1}_i}}{\text{ESS}_{n+1}}=\frac{\|\hat\pi_{n+1,K_{n+1}}\|}{\text{ESS}_{n+1}}.
\end{eqnarray*}
On the other hand, by applying Theorem~$\ref{consistency}$ for $\phi \equiv h$, we have
\begin{eqnarray*}
\bar{\pi}_{n+1,K_{n+1}}(h) \h  \to \h \bar{\pi}_{n+1}(h) = \frac{\|\hat \pi_{n+1}\|}{\|\hat \pi_{n}\|} \int_{t \in \mathcal{T}_{n+1}}{\frac{\bar \pi_{n+1}^2(t)}{\bar \pi_{n}(\varrho(t))~Q^n(\varrho(t),t)}~dt}
\end{eqnarray*}
%Note for EM: just invert (4.2) and take limit of product.
which completes the proof via the convergence result $\eqref{eqess}$.
\end{proof}

\begin{proof}[Proof of Lemma \ref{subtree}]
Let $l_e$ be the length of the edge $e$ and $G(\alpha)$ be the transition matrix across an edge of length $\alpha$ and $k_u$ the observed value at site $u$ of the newly added taxon.
We follow the formulation of one-dimensional phylogenetic likelihood function as in \citep{dinh2015shape} to fix all parameters except $l_e$  and consider the likelihood of $\varrho(t)$ a function of $l_e$, we have
\[
L_{n}(\varrho(t)) =  \prod_{u=1}^S{\left(\sum_{ij}{d^{u}_{ij}G^{e}_{ij}(l_e)}\right)}
\]
where $d^{u}_{ij}$ the probability of observing $i$ and $j$ at the left and right nodes of $e$ at the site index $u$, respectively (note that in \citep{dinh2015shape} it is called $d^{u}_{ij}$).
Similarly, by representing the likelihood of the tree $t$ in terms of $x$, $y$ and $l_e$, we have
\begin{equation}
L_{n+1}(t) =  \prod_{u=1}^S{\left(\sum_{ij}{d^{u}_{ij}G^{e}_{ij}(l_e,x,y)}\right)}
=\prod_{u=1}^S{\left(\sum_{ij}{d^{u}_{ij}  \sum_{m}{G_{im}(x) G_{mj}(l_e-x) G_{mk_u}(y)}  }\right)}
\label{eq:splitLike}
\end{equation}
where the indices $i,j,m$ are looped over all possible state characters.
Since $G_{mk_u}(y) \le M(y)$ for all $m$ and $k_u$, we deduce that
\begin{equation}
\frac{\hat \pi_{n+1}(t)}{\hat \pi_n (\varrho(t))}=\frac{L_{n+1}(t)}{L_n (\varrho(t))} \le M(y)^S, \h \forall t \in \mathcal{T}_{n+1}.
\label{util1app}
\end{equation}
Similarly, we have $\hat \pi_{n+1}(t)/\hat \pi_n (\varrho(t)) \ge m(y)^S$ for all $t \in \mathcal{T}_{n+1}$.

Recall that $\zeta(r)$ is the average branch length of $r$.
Using the fact that for a fixed tree $r$, $\int_0^{l_e}{dx = l_e}$ and $\sum_{e}{l_e} = (2n-3) \zeta(r)$, we have
\begin{align*}
\|\hat \pi_{n+1}\| =& \int_{t \in \mathcal{T}_{n+1}}{\hat \pi_{n+1}(t) \, dt}
\ge \int_{t \in \mathcal{T}_{n+1}}{m(y)^S \hat \pi_{n}(\varrho(t)) \, dt}\\
=& \int_{r,e,x,y}{m(y)^S \hat\pi_n(r) \, dx \, dy} ~ \frac{V_n}{V_{n+1}}de \, dr\\
=& \frac{(2n-3)V_n}{V_{n+1}} \left(\int_0^{b}{m(y)^S \, dy}\right)\int_{r \in \mathcal{T}_{n}}{{\hat \pi_{n}(r)}\zeta(r) \, dr}.
\end{align*}

Noting that $V_{n+1}=(2n-3)V_n$, we obtain
\begin{equation}
\frac{\|\hat \pi_{n+1}\|}{\|\hat \pi_n\|} \ge\left(\int_0^{b}{m(y)^S~dy}\right)\mathcal{Z}_n
\label{eq-1}
\end{equation}
which implies
\[
\frac{\bar \pi_{n+1}(t)}{\bar \pi_n (\varrho(t))} = \frac{\hat \pi_{n+1}(t)}{\hat \pi_n (\varrho(t))}  \frac{\|\hat \pi_{n}\|}{\|\hat \pi_{n+1}\|} \le \frac{1}{\mathcal{Z}_n} \frac{M(y)^S}{\int_0^{b}{m(y)^S~dy}}, \h \forall t \in \mathcal{T}_{n+1}.
\]
\end{proof}

\begin{proof}[Proof of Theorem $\ref{lengthbased}$]

Since the edge $e$ is chosen from a multinomial distribution weighted by length of the edges, then given any tree $t \in \mathcal{T}_{n+1}$ obtained from the parent tree $\varrho(t)$ by choosing edge $e$, distal position $x$ and pendant length $y$, we have
\[
Q^n(\varrho(t),t) = \frac{V_{n+1}}{V_n} \frac{l_e(\varrho(t))}{l(\varrho(t))}~ p_X(x) ~ p_Y(y)
\]
where $l_e(r), l(r)$ are the length of edge $e$ and the total tree length, respectively and $V_n$ and $V_{n+1}$ are the numbers of tree topologies of $\mathcal{T}_n$ and $\mathcal{T}_{n+1}$.

We denote
\[
u_1 = \int_0^{b}{m(y)^S~dy}, \h u_2 =\int_0^{b}{ \frac{M(y)^{2S}}{p_Y(y)}~dy},
\]
and recall that
\[
\sum_{e}{\frac{1}{l_e(r)} \int_0^{l_e(r)}{\frac{1}{p^e_X(x)}~dx}} ~ \le ~ C \sum_{e}{l_e(r)} ~= ~ C l(r),
\]
where $C$ is the constant from Assumption $\ref{distal}$, and
\[
\mathcal{Z}_{n}=\int_{s \in \mathcal{T}_{n}}{{\bar\pi_{n}(r)}\zeta(r)~dr} \ge c
\]
from the assumption on the average branch length (Assumption $\ref{branchlength}$).
We have
%EM The arrangement of text in this following math section looks funny (first line pushed too far right). I think the "split" environment is what's needed here-- perhaps give it a try?
%V I changed the environment to "align", which is more good looking. I will try "split" later.
\begin{align*}
&\int_{t \in \mathcal{T}_{n+1}}{\frac{\bar \pi_{n+1}^2(t)}{\bar \pi_{n}(\varrho(t))~Q^n(\varrho(t),t)}~dt}
= \int_{t \in \mathcal{T}_{n+1}}{\frac{\bar \pi_{n+1}^2(t)}{\bar \pi^2_{n}(\varrho(t))} \frac{1}{Q^n(\varrho(t),t)}~\bar \pi_{n}(\varrho(t)) ~dt} \\
&\le\left(\frac{V_n}{V_{n+1}}\right)^2  \frac{1}{\mathcal{Z}_n^2} ~ \int_{r,e,x,y}{ \frac{M(y)^{2S}}{u_1^2} \frac{l(r)}{l_e(r)} \frac{1}{p^e_X(x)}  \frac{1}{p_Y(y)}\bar \pi_{n}(r) ~dx~dy~de~dr}\\
&=  \left(\frac{V_n}{V_{n+1}}\right)^2 \frac{1}{\mathcal{Z}_n^2} \frac{u_2}{u_1^2} ~\int_{\mathcal{T}_n} { \left(\sum_{e}{\frac{1}{l_e(r)} \int_0^{l_e(r)}{\frac{1}{p^e_X(x)}~dx}}\right)~l(r)~\bar \pi_{n}(r) ~ dr}.
\end{align*}

By the assumption of maximum branch length $b$, we have
\begin{align*}
\int_{t \in \mathcal{T}_{n+1}}{\frac{\bar \pi_{n+1}^2(t)}{\bar \pi_{n}(\varrho(t))~Q^n(\varrho(t),t)}~dt} &\le C  (2n-3)^2 \left(\frac{V_n}{V_{n+1}}\right)^2 \frac{1}{\mathcal{Z}_n^2} \int_{\mathcal{T}_n}{\bar \pi_n(r)\zeta^2(r)~dr} \\
&\le \frac{Cb^2}{c^2}.
\end{align*}
Thus by Theorem~\ref{thm:lengthESSLimit} there exists $\alpha>0$ independent of $K_n$ and $n$ such that $\text{ESS}_{n} \ge \alpha K_{n}$.
\end{proof}

\end{document}